\documentclass[12pt,letterpaper]{article}

\usepackage{fixltx2e}
\usepackage{textcomp}
\usepackage{fullpage}
\usepackage{amsfonts}
\usepackage{verbatim}
\usepackage[english]{babel}
\usepackage{pifont}
\usepackage{color}
\usepackage{setspace}
\usepackage{lscape}
\usepackage{lscape}
\usepackage{indentfirst}
\usepackage[normalem]{ulem}
\usepackage{booktabs}
\usepackage{natbib}
\usepackage{float}
\usepackage{latexsym}
\usepackage{url}
\usepackage{hyperref}
\usepackage{epsfig}
\usepackage{graphicx}
\usepackage{amssymb}
\usepackage{amsmath}
\usepackage{bm}
\usepackage{array}
\usepackage{ifthen}
\usepackage{caption}
\usepackage{hyperref}
\usepackage{amsthm}
\usepackage{amstext}

\usepackage{epic}

\usepackage{latexsym}
\usepackage{amsmath}
\usepackage{epsfig}
\usepackage{amssymb}
\usepackage{enumerate}

\usepackage[mathscr]{euscript}

\newtheorem{theorem}{Theorem}
 
\newtheorem{corollary}[theorem]{Corollary}
\newtheorem{proposition}[theorem]{Proposition}

\newcommand{\PP}{{\mathbb P}}

\newcommand{\cQ}{{\mathcal Q}}

\newcommand{\cS}{{\mathcal S}}
\newcommand{\cL}{{\mathcal L}}
\newcommand{\cT}{{\mathcal T}}
\newcommand{\cP}{{\mathcal P}}

\newcommand{\cM}{{\mathcal M}}

\raggedright
\renewcommand{\section}[1]{%
\bigskip
\begin{center}
\begin{Large}
\normalfont\scshape #1
\medskip
\end{Large}
\end{center}}

\renewcommand{\subsection}[1]{%
\bigskip
\begin{center}
\begin{large}
\normalfont\itshape #1
\end{large}
\end{center}}

\renewcommand{\subsubsection}[1]{%
\vspace{2ex}
\noindent
\textit{#1.}---}

\renewcommand{\tableofcontents}{}

\bibpunct{(}{)}{;}{a}{}{,}  

\begin{document}
\begin{flushright}
Version dated: \today
\end{flushright}
\bigskip
\noindent RH:  IMPACTS OF TERRACES

\bigskip
\medskip
\begin{center}

\noindent{\Large \bf Impacts of terraces on phylogenetic inference}
\bigskip



\noindent {\normalsize \sc Michael J. Sanderson$^{1,6}$, Michelle M. McMahon$^2$, Alexandros Stamatakis$^{1,3,4}$, Derrick J. Zwickl$^1$,  and Mike Steel$^{5,6}$, }\\
\noindent {\small \it 
$^1$Department of Ecology and Evolutionary Biology, University of Arizona, Tucson, AZ, 85721, USA;\\
$^2$School of Plant Sciences, University of Arizona, Tucson, AZ, 85721, USA;\\
$^3$Scientific Computing Group, Heidelberg Institute for Theoretical Studies, Heidelberg, 69118, Germany;
$^4$Institute of Theoretical Informatics, Karlsruhe Institute of Technology, Karlsruhe, 76131, Germany;
$^5$Biomathematics Research Centre, University of Canterbury, Christchurch, NZ;\\
$^6$Authors contributed equally to work}\\
\end{center}
\medskip
\noindent{\bf Corresponding author:} Michael J. Sanderson, Department of Ecology and Evolutionary Biology, University of Arizona, Tucson, AZ, 85721, USA; E-mail: sanderm@email.arizona.edu.\\


\vspace{1in}

\newpage
\subsubsection{Abstract} Terraces are potentially large sets of trees with precisely the same likelihood or parsimony score, which can be induced by missing sequences in partitioned multi-locus phylogenetic data matrices. The set of trees on a terrace can be characterized by enumeration algorithms or consensus methods that exploit the pattern of partial taxon coverage in the data, independent of the sequence data themselves. Terraces add ambiguity and complexity to phylogenetic inference particularly in settings where inference is already challenging: data sets with many taxa and relatively few loci. In this paper we present five new findings about terraces and their impacts on phylogenetic inference. First we clarify assumptions about model parameters that are necessary for the existence of terraces. Second, we explore the dependence of terrace size on partitioning scheme and indicate how to find the partitioning scheme associated with the largest terrace containing a given tree. Third, we highlight the impact of terraces on bootstrap estimates of confidence limits in clades, and characterize the surprising result that the bootstrap proportion for a clade can be entirely determined by the frequency of bipartitions on a terrace, with some bipartitions receiving high support even when incorrect. Fourth, we dissect some effects of prior distributions of edge lengths on the computed posterior probabilities of clades on terraces, to understand an example in which long edges ``attract" each other in Bayesian inference. Fifth, we show that even if data are not partitioned, patterns of missing data studied in the terrace problem can lead to instances of apparent statistical inconsistency when even a small element of heterotachy is introduced to the model generating the sequence data. Finally, we discuss strategies for remediation of some of these problems. Among the most promising is the strategic deletion of a minimal number of taxa from the data in order to reduce terrace sizes.

\noindent (Keywords: phylogenetics, terrace, partitioned model, bootstrap, posterior probability)\\


\vspace{1.5in}

\newpage

Phylogenetic trees with thousands to tens of thousands of species are becoming increasingly commonplace \citep{rab13,zan14}. They can serve at least two purposes: quantifying and conveying the true scale and breadth of biodiversity, and providing statistical power to distinguish between alternative models of evolution \citep{wie11a,boe12,cha12,gol12,mar12,chr13,dav13}. Reconstruction of large trees entails many challenges \citep{san07,izq11,liu12}, including a recently discovered one: ``terraces" \citep{san11}. A terrace is a region in tree space in which all trees have precisely the same likelihood and parsimony score, which adds ambiguity to the ``landscape" of trees (Fig.~\ref{fig1}) and complexity to tree inference. Terraces may have been overlooked among the inevitable small numerical differences that arise in computing the likelihood score on different trees, especially in large data sets. Indeed, the issue of whether two trees have {\em exactly} the same likelihood rarely arises in modern phylogenetic inference for this reason. 

	\begin{figure}[htb]
\centering
\includegraphics[scale=0.66]{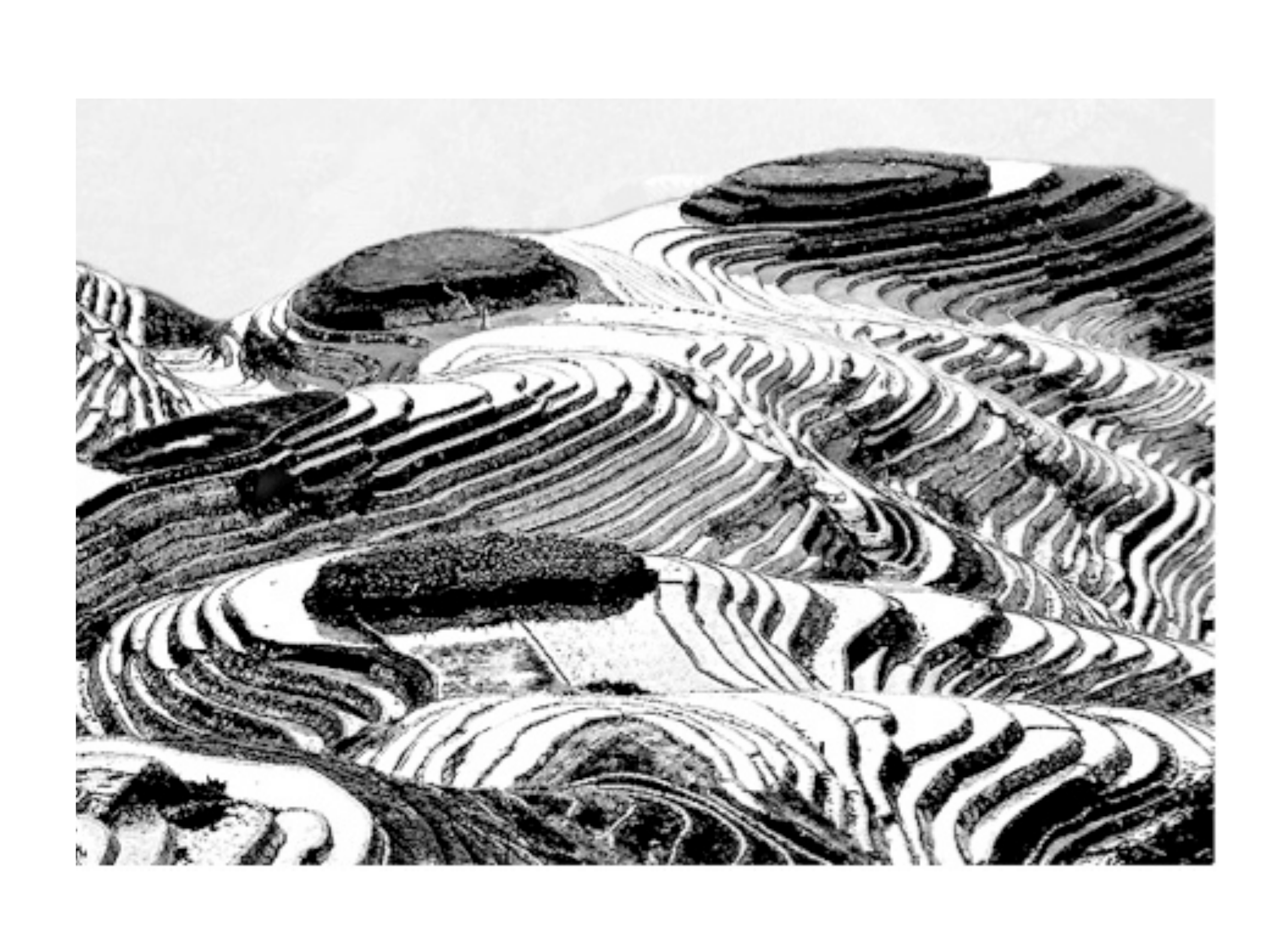}
\caption{Schematic view of terraces in a likelihood surface of phylogenetic trees.}
\label{fig1}
\end{figure}

In trying to improve the efficiency of likelihood calculations in RAxML, \citet{StamAlAch10} pointed out one important context in which likelihoods of different trees could be precisely identical: when the subtrees for different loci in a multi-locus partitioned analysis are the same. We suggested the term ``terrace" for the resulting trees and exploited a variety of results concerning subtrees and supertrees to characterize these terraces, which in some data sets can be quite large \citep{san11}. In general, trees with equal or nearly equal optimality scores can arise for many reasons, including lack of variable sites, homoplasy, and missing data. Terraces can arise in parsimony analysis when there are missing data or in likelihood-based analyses when missing data are distributed differently in different elements or blocks of a specified data partition. 

Data may be missing for many reasons, ranging from sampling biases inherent in studies that mine GenBank \citep{dri04,san08} to more biological causes, as in the loss of plastid genes transferred to the nuclear genomes of some plants \citep{Sabir14}, or the differential expression of genes found in EST libraries or transcriptomes \citep{let12}. Missing data affect phylogenetic tree reconstruction in many ways \citep{kea02,bur09,lem09,cho11,SimFreu11, wie11b,cra12,Sim12a,Sim12b,SimGol13, hin13,rou13,Sim14a,Sim14b}, but many properties of terraces depend only on the pattern of ``taxon coverage," the set of taxa for which data are present, in the different elements of the data partition \citep{ste10,san10,san11}. 

	In practice, terraces are generally small or absent in phylogenomic studies in which the number of loci greatly exceeds the number of taxa (\citealp{hej09}: 94 taxa $\times$ 1487 loci; \citealp{sal13}: 23 taxa $\times$ 1070 loci; \citealp{zwi14}: 11 taxa $\times$ 473 loci), but when the number of taxa greatly exceeds the number of loci, the number of trees on each terrace can be extremely large and the resulting increase in ambiguity challenging to overcome  (e.g., \citealp{pyr11}: 2871 taxa $\times$ 12 loci; \citealp{smi09}:
	55,473 taxa $\times$ 6 loci; \citealp{fab12}: 1265 taxa $\times$ 11 loci; \citealp{rab13}: 7822 ray-finned fish taxa $\times$ 13 loci).  Building on several basic mathematical properties of terraces we have characterized quantitatively \citep{san11}, we extend our understanding of terraces in several directions, examining several new properties, the problems they induce, and strategies for overcoming them. We show that terraces can arise under more general conditions than we thought, and that they can be larger than believed before. We construct for parsimony a method to recover the maximal terrace for a given tree. Most importantly, we explore how terraces can affect confidence assessments and conventional views on how methods for characterizing ambiguity, such as consensus methods or bipartition support values, can be misled by them. Finally, we begin to extend results to the case in which the model partitioning scheme violates the sufficient conditions for terraces but still generates data sets with patterns of ambiguity related to our other results on terraces.

\section{Background}	
\subsection{Definitions}
\subsubsection{Data}
Let $D$ be a data matrix, such as a multiple sequence alignment, of $n$ taxa (``rows") and $l$ characters (``sites", ``columns"), and let $\cP$ be a partition of the columns into $m$ elements or blocks, referred to colloquially throughout this paper as ``loci," although blocks might be something like different sets of codon positions, etc. Let the {\em taxon coverage matrix}, $C= C(D$), be an $n \times m$ matrix where the $ij$-th element is 1 if any sequence data are present for taxon $i$ and locus $j$, and 0 if the data are entirely missing (Fig.~\ref{fig2}).  This often happens simply because the locus was not sampled for that taxon. Note that there is no guarantee that taxa scored as ``1'' for a locus have phylogenetically {\em informative} sequence data--they merely have sequence data of some kind. 

	\begin{figure}[!h]
\centering
\includegraphics[scale=0.66]{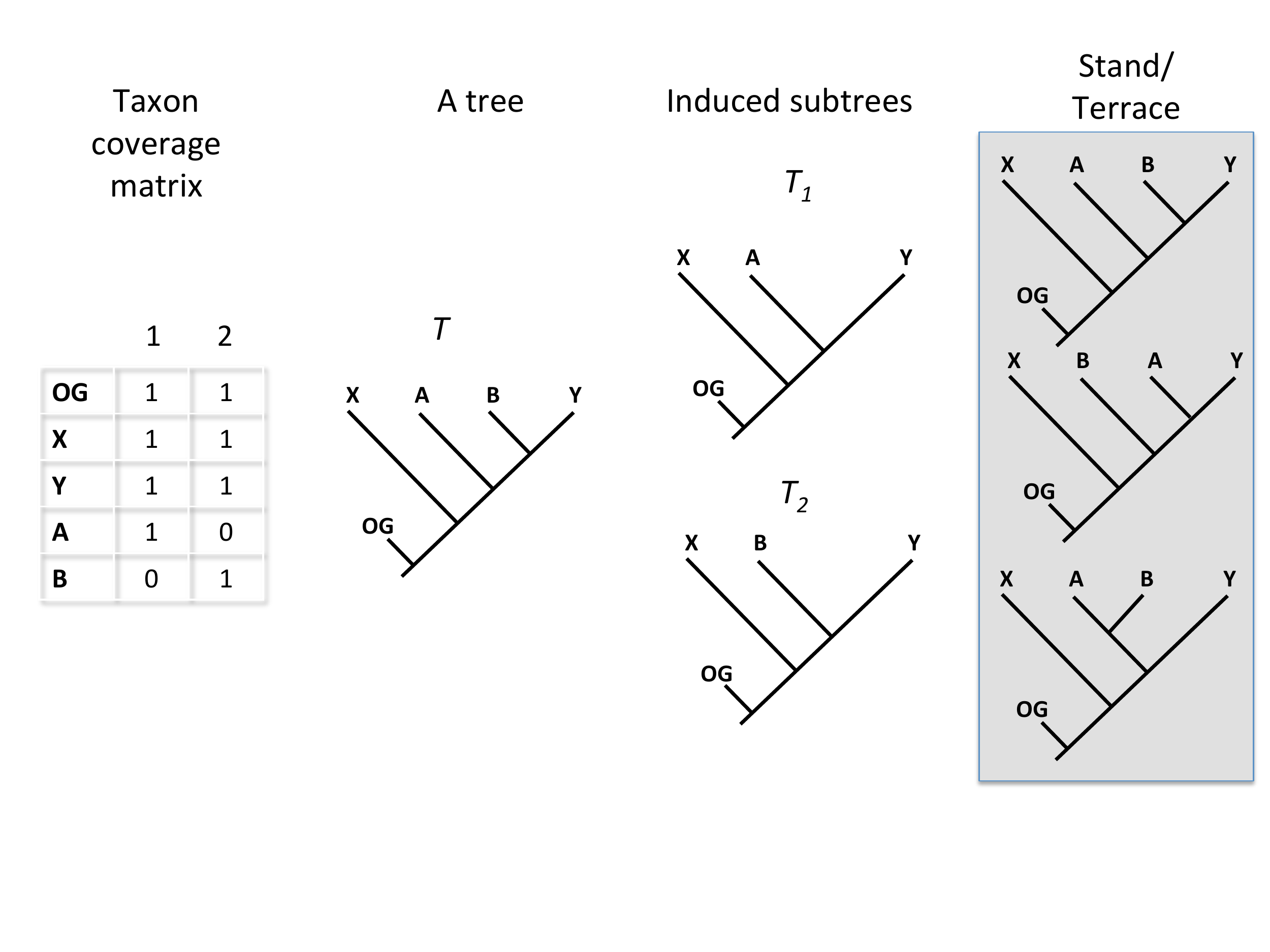}
\caption{Stands and terraces. The two-locus taxon coverage matrix at left (? indicates missing data) is partially filled and not decisive. For at least one tree (e.g., $T$), and the two subtrees induced by the coverage matrix, there is more than one tree (the three at right) consistent with the subtrees. This collection of trees is a {\em stand}. In addition, because the parsimony score and likelihood score (with an unlinked model), will be identical on these three trees, the stand is also a {\em terrace}.}
\label{fig2}
\end{figure}

The taxon {\em label set} for $D$, $\cL$, is the set of all taxon names in $D$, and the label set, $\cL_j$, for block $j$ of $D$ is the set of taxon names for which data are present for block $j$ in partition $\cP$ (i.e., the set of taxon names for which $C_{ij}$ = 1). 

\subsubsection{Trees and subtrees}For any tree, $T_0$, on the complete label set $\cL$ (e.g., the best maximum likelihood tree found in a heuristic tree search based on $D$), the label set for any locus, $j$, $\cL_j$ (which may have taxa missing) induces a subtree of $T_0$, which we write as $T_0 | \cL_j$. Let $\cQ(\cP,C,T_0) = {T_0 | \cL_1, T_0 | \cL_2,...,T_0 | \cL_m}$, be the set of subtrees for this tree that are induced by the partitioning scheme and taxon coverage matrix--that is, $T_0|\cL_j$ is the tree obtained from $T_0$ by removing any taxa that have data that are entirely missing for block $\cL_j$, for each block $\cL_j$ of the partition.

A tree $T$ is a {\em resolution} of $T'$ if $T'$ can be obtained from $T$ by collapsing one or more edges of $T$, which transforms binary nodes to polytomies. A tree $T$ on label set $\cL(T)$ {\em displays} another tree $T'$ on $\cL(T')$ (with $\cL(T') \subseteq \cL(T)$) if $T'$ is equal to or a resolution of $T|\cL(T')$. Intuitively, this allows a larger tree to display a smaller tree even if that smaller tree is less resolved than the larger one. This extra generality is appropriate given the usual notion of polytomies as reflecting uncertainty rather than multiple speciation \citep{sem03b}. 

\subsubsection{Stands}Here we define a new term not used in our previous work on terraces. A {\em stand} (of trees), $\cS=\cS(\cP,C,T_0)$, is the set of all binary phylogenetic trees on leaf set $\cL$ that display every subtree in $\cQ(\cP,C,T_0)$. In the example in Figure~\ref{fig2}, there are three trees that display the subtrees in $\cQ$, including the original $T_0$.
	
\subsubsection{Decisiveness}The taxon coverage matrix is said to be {\em decisive} for tree $T_0$ and partitioning scheme $\cP$ if $T_0$ is the only  tree that displays all subtrees in $\cQ$ \citep{ste10,san10}
--in other words,  if the stand has only one tree (i.e. $|\cS|=1$). In the example in Figure~\ref{fig2}, the taxon coverage matrix is not decisive for $T_0$.

\subsubsection{Terraces}Stands and decisiveness depend only on the coverage pattern, partition scheme and tree,  $\{\cP,C,T_0\}$; not on the actual sequence data, $D$, per se. However, under certain conditions, all trees in a stand, $\cS$, have precisely the same optimality score with respect to $D$, in which case we call the stand, $\cS$, a {\em terrace}, denoted $\cT$ (Fig. 2). In particular, when parsimony is used as the score, all trees in $\cS$ have the same score for {\em any} partitioning scheme, $\cP$, and $\cS$ is always a terrace. The set $\cS$ is also a terrace when likelihood is the optimality criterion if the likelihood is determined by a model for each locus in the partition, $\cP$, that is ``unlinked'' \citep{san11}. An unlinked (also known as ``partitioned'', or ``fully partitioned'' in the literature) model has different sets of parameters, including edge lengths, for each locus \citep{hes11,hed12,xiz13}. 

\subsection{Properties of terraces}
Because all terraces are stands, a number of results derived for stands, based only on trees, subtrees, and coverage patterns, are helpful for characterizing terraces. Some of these results hold only for rooted trees, but many problems can be effectively ``rooted" as long as there is one taxon in $C$ that is sampled for all blocks of the partition, in which case that taxon can serve as an ``operational" root for the purposes of an algorithm. The number of trees on a stand (or terrace) can increase exponentially with the size of the tree \citep{sem03a}. Despite this, several properties of terraces make them more tractable than they might otherwise seem, mainly because several summary statistics can be obtained directly without any computation involving the data matrix, $D$  \citep{san11}. For example, all trees on a terrace can be enumerated without recalculation of optimality scores. This takes advantage of an algorithm due to \citet{con95}, with a running time that scales linearly with the size of the terrace (rather than, say, exponentially with the size of the tree). 

The trees on a terrace can also be summarized by a strict consensus tree \citep{gor86} or Adams consensus tree, either of which can be constructed in polynomial time. This last claim (that we can  sidestep the enumeration of trees on the terrace again, or any further search using the data) is not obvious, but it holds in general (for rooted trees). In the case of strict consensus this was shown in  \citep{ste92} (using results from \citet{aho81}), while for Adams consensus, it relies on a particularly elegant result due to  \cite{bry} (Theorem 6.2) which states that the Adams consensus of a terrace is equal to the so-called BUILD supertree of the set of subtrees for each locus (i.e., the set $\cQ(\cP,C,T_0)$), and this supertree can be computed quickly by the algorithm of 
\citet{aho81}. Note that the Adams consensus tree displays each of those subtrees; that is, it is identical to (or resolves) them if extraneous taxa and edges are removed. Finally, testing if two trees, $T_1$ and $T_2$, are on the same terrace can be done efficiently merely by checking for equality of the pair of induced subtrees, $T_1 | \cL_j  = T_2 | \cL_j$; for every locus $j$ (cf. \citealp{day85}).
	
Terraces are reminiscent of tree islands but are contained within them, at least for rooted trees. An island is a region of tree space with optimality score better than some threshold, separated from other such regions by regions of lower score \citep{mad91,sal01}. Here a ``region'' is a set of trees that can be enumerated by a series of topological rearrangements that do not leave the region. Terraces are always wholly contained within a tree island, because all trees on a terrace of rooted trees can be reached by a series of nearest neighbor interchanges (NNIs) between trees of the same optimality score \citep{bor03,san11}. For unrooted trees this property does not necessarily hold.  For rooted trees, then, tree space can be thought of as a collection of terraces (on islands), the size and topological position of which are determined entirely by the partition $\cP$ and taxon coverage matrix $C$. Only the heights of the terraces depend on the data. To the extent that a tree ``island'' is a useful metaphor, it may be best to envision it as a rough landscape covered with terraces of different sizes and heights (Fig.~\ref{fig1}).

\section{New Results on Terraces}	

\subsection{1. Terraces occur in likelihood inference under a less restrictive set of assumptions}
In our previous work, we showed that terraces can occur in maximum likelihood inference whenever the substitution models for separate loci in partition $\cP$ have separate parameter sets \citep{san11}. Such models are known as ``partitioned'' or ``unlinked'' models. To be precise, suppose the substitution model at locus $i$ consists of two sets of free parameters $\{M_i, E_i\}$, where $M_i$ is a set of free parameters associated with the substitution rate matrix alone, and $E_i$ is a set of free edge rate parameters. An  {\em edge-unlinked} (EUL) model has free edge rate parameters $E_i$ for each locus. On the other hand if any  constraints are imposed to reduce the number of free edge parameters (such as those of the form $E_i \propto E_j$, for at least some $i$, $j$), we will say the model is {\em edge-linked} (EL); likewise for  substitution matrix rate-linked (RL) and rate-unlinked (RUL) models. Later we will refer to two simple EL models: the ``homotachy" model, in which $E_i = E_j$ for all $i$, $j$, and the ``proportional" model in which $E_i = \alpha E_j$ and $\alpha$ varies across loci to reflect locus specific overall rates.

In terms of these definitions, we showed \citep{san11} that terraces can arise when models are simultaneously EUL and RUL. This was a sufficient but not necessary condition. 
In fact, as we show now, terraces can arise even when models are EUL and the substitution rate matrices are constrained to be the same across loci.

\begin{proposition}
\label{thm1}
Let $\cS(C,\cP,T_0)$ be a stand of trees containing tree $T_0$ for partitioning scheme, $\cP$ and taxon coverage matrix, $C$.  For any model $\cM$ for computing the likelihood score  that is edge-unlinked (EUL)  and where  the $M_i$ parameters are constrained to be identical (i.e. $M_i = M$ for all $i$), all trees in $\cS$ have the same maximum likelihood score, $\varphi(\cM,T_0|D)$, as $T_0$, and hence $\cS$ is a terrace.
\end{proposition}

The proof of this result is presented in the Appendix. This finding implies that it is the lack of commonality between {\em edge length parameters} across loci that is necessary for a stand of trees to be a terrace under likelihood. Rate matrices between loci need not have different parameter sets. Henceforth we will refer to this broader class of maximum likelihood inference with just an EUL assumption as ``ML-EUL.''

\subsection{2. Maximum size of stands and terraces}

Given a coverage matrix, $C$ and partition scheme, $\cP$,  we can calculate the size of the stand, $\cS$ containing some tree, $T_0$, without reference to the data, $D$. If the optimality criterion is MP or ML-EUL, then $\cS$ is also a terrace and the size of the terrace will be the size of the stand. Generally, the presumption is that $\cP$ is chosen to reflect meaningful biological aspects of the data, such as different loci, or disjoint sets of codon positions. Nonetheless, the choice is arbitrary and it is possible that this terrace for $\cP$ is actually imbedded in a larger terrace (perhaps with a different likelihood score) under a different partitioning scheme. First, we define a special kind of partitioning scheme and then show that this will correspond to the largest terrace among all possible partitioning schemes. 

Given a taxon coverage matrix, $C$, consider any partition $\cP = \{B_1,\ldots, B_l\}$ of the columns of the data matrix, $D$ (i.e.  the sets $B_i$ are disjoint subsets of columns, which cover every column).  For a given block $B_i$ of $\cP$, let $\cL_i = \cL(B_i)$ be the taxa that are present for at least one column in $B_i$.

Now suppose we have another partition $\cP' = \{B_1', \ldots, B'_{l'}\}$.  
We say that $\cP'$ {\em refines} $\cP$ if each block of $\cP'$ is a subset of some block of $\cP$ (equivalently, each block of $\cP$ is either equal to a block of $\cP'$ or is the union of
two or more blocks of $\cP'$). 

However, a refinement that includes breaking $B_j$ into say, $K$, smaller sets,  $B'_{i_1}, \ldots, B'_{i_K}$, should be disallowed if any of the labels sets, $\cL(B'_{i_k})$, are duplicated, as otherwise this would allow a trivial refinement of any partition into one in which each block consists of just a single column in the data matrix, $D$. With this restriction in place, there will be a unique (and usually nontrivial) {\em maximal partition} for any taxon coverage matrix, $C$, that cannot be refined further, we denote this as $\cP_{\rm max}$.  Without this restriction the maximal partition would simply consist of  $l$ blocks, one for each column in the data matrix.  This maximal permitted partition $\cP_{\rm max}$ can be described formally as follows. Let $\Omega = \{\cL_1, \ldots, \cL_m\}$ (this set may have size less than $m$ since some $\cL_i=\cL_j$ is possible for $i \neq j$).  Then $\cP_{\rm max}= (B_A: A \in \Omega)$, where $B_A = \{j: \cL_j = A\}$.

\begin{proposition}
\label{pro1}
If $\cP'$ refines $\cP$ then $\cS(C,\cP,T_0)  \subseteq\cS(C,\cP',T_0)$; in  particular, the former set is never larger than the latter.
\end{proposition}
\begin{proof}
Let $\cL'_j = \cL(B'_j)$, the taxa present in block $B'_j$ for $j = 1, \ldots, l'$. 
Suppose that $T$ is a tree in the stand of $T_0$ relative to $\cP$.  This means that
$T|\cL_i = T_0|\cL_i$ for $i =1, \ldots, l$.  Consider a block $B'_j$ of $\cP'$. Since $\cP'$ refines $\cP$,   $B'_j$  is a subset of some block of $\cP$, say block $B_i$, and so  $\cL'_j \subseteq \cL_i$. 
In that case:
$$T|\cL'_j = (T|\cL_i)|\cL'_j = (T_0|\cL_i )|\cL'_j = T_0|\cL'_j.$$
Thus, $T|\cL'_j = T_0|\cL'_j$ for each block $\cL'_j$ of $\cP'$, and  so $T$ is in the stand of $T_0$ relative to $\cP'$.
\end{proof}

\begin{corollary}
For any permitted partition $\cP$ we have:
$|\cS(D, T_0, \cP_{\rm max})| \ge |\cS(D, T_0, \cP)|$
for all $\cP$.
\end{corollary}
\begin{proof}
This follows from Proposition~\ref{pro1}, since $\cP_{\rm max}$ is a refinement of every permitted partition of the loci. 
\end{proof}

	Thus, the largest stand containing a given tree can be obtained from the maximal partitioning scheme allowed for a given taxon coverage pattern, $C$. Under MP or ML-EUL inference, this will also be the largest terrace containing this tree. 
	
	For example, suppose a multiple sequence alignment is first partitioned into two loci, $\cP_I = \{I_1, I_2\}$, where the loci are labeled by the partition name with a subscript for the locus number, as shown in Figure~\ref{fig3}. This partition induces a stand of 13 trees containing the tree shown. However, perhaps locus $I_2$ actually consists of three biologically meaningful blocks: two exons and an intron, and suppose further that the label sets for the exons are the same but differ from that for the intron. Our second partition is then described by $\cP_{II} = \{II_1, II_2, II_3\}$ (and has the indicated distribution of columns: Fig. ~\ref{fig3}). Formally, $\cP_{II}$ is a refinement of $\cP_{I}$. The stand induced by this new three block partition, $\cP_{II}$, has 23 trees (Figure ~\ref{fig3}). This is the maximal partition for these data, and 23 is the size of the largest stand containing the given tree. Also, from Proposition \ref{pro1}, the smaller stand of trees is a subset of the larger.
	
\begin{figure}[!h]
\centering
\includegraphics[scale=0.66]{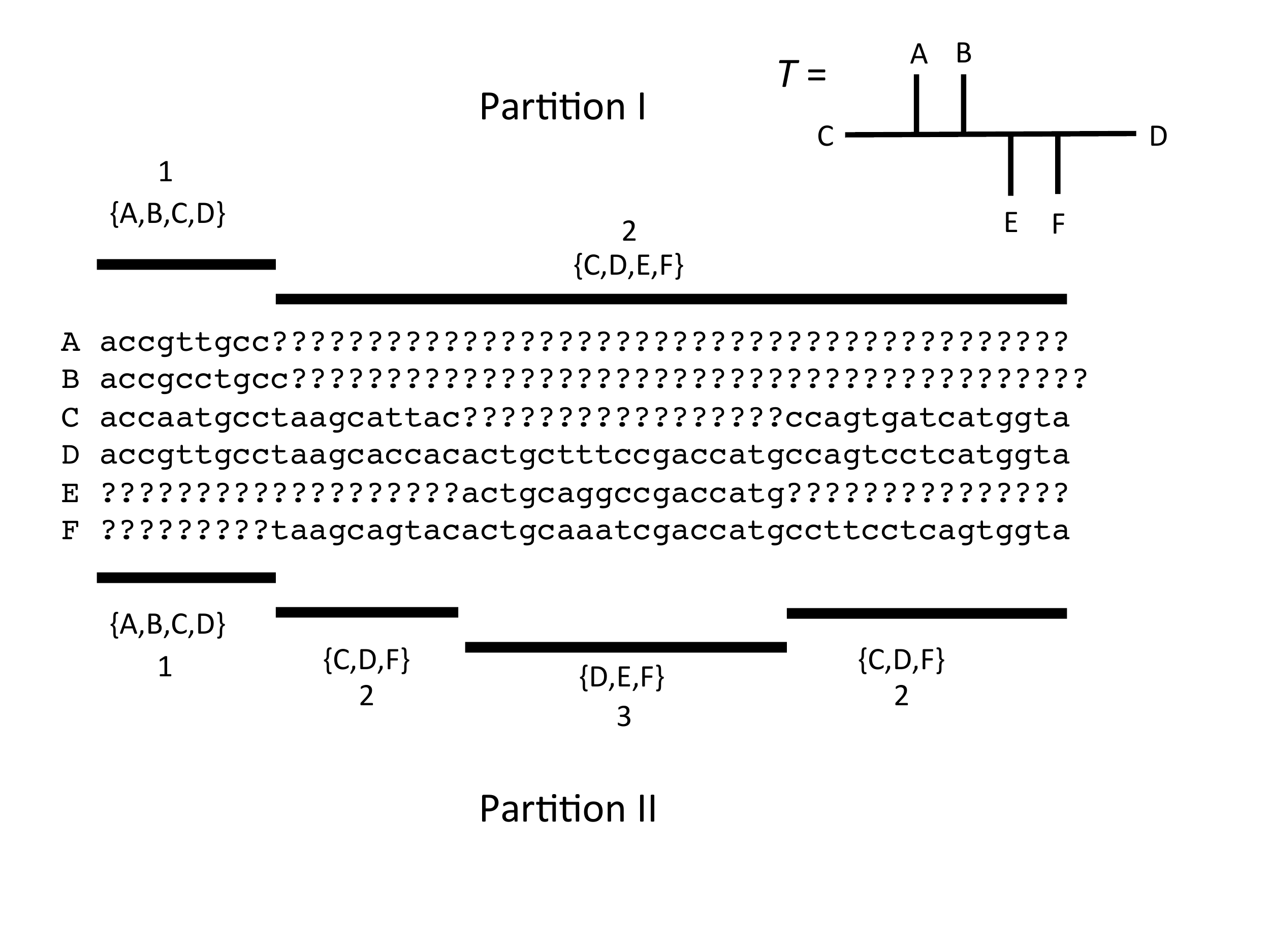}
\caption{Partitions of different sizes and compositions can induce terraces of different sizes containing a given tree, $T$ (upper right). Two partitioning schemes are labeled I and II. Partition elements/blocks/``loci" are labeled with 1,2, 2(1), etc. Label sets for each block are indicated in curly brackets. Partition II is a ``maximal partition". There are 13 trees on the terrace for Partition I and 23 trees for Partition II (for parsimony), each containing $T$.}
\label{fig3}
\end{figure}

	Under MP and ML-EUL inference, these stands are also terraces, so the largest terrace containing the indicated tree has 23 trees on it. The sets of trees in these terraces arising from different partitioning schemes are the same for MP and ML-EUL inference, but there is an interesting distinction with respect to their optimality scores. For MP, the situation is particularly simple. For two partitioning schemes in which $\cP_{II}$ is a refinement of $\cP_{I}$, the two terraces have the same parsimony score, so it is reasonable to visualize the smaller terrace is being literally imbedded in the larger one at the same ``elevation" in the landscape of tree space (Figure ~\ref{fig4}). Consequently, for parsimony, when characterizing trees on a terrace based on a particular partition, $\cP$, it is useful to check if $\cP$ is the maximal partition, and if not, also check $\cP_{\rm max}$. The latter will be a more accurate representation of the actual extent of equally optimal trees around the given tree. 
	
	For ML-EUL, the situation is more complex. The likelihood scores of the trees in the terraces for partitioning schemes $\cP_{I}$ and $\cP_{II}$ could well be different, but all trees within either terrace will still have the same score once a partitioning scheme is fixed. This means changing a partitioning scheme to the maximal partitioning scheme, for example, not only expands the set of trees on the terrace, but potentially changes the likelihood score as well, unlike in parsimony.
\begin{figure}[!h]
\centering
\includegraphics[scale=0.66]{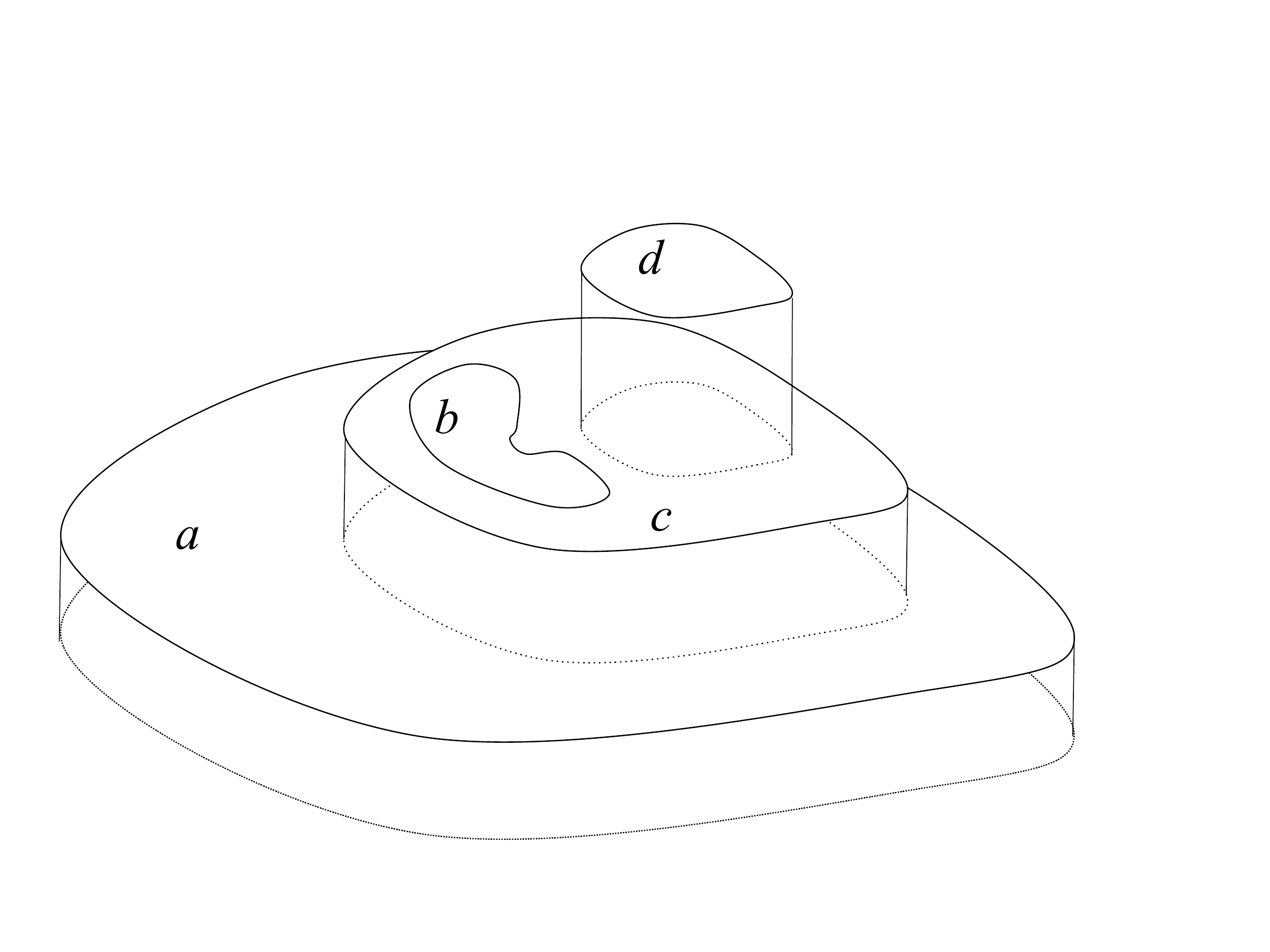}
\caption{Refined view of the landscape of terraces with respect to different partitioning schemes. Regions 1-4 are different regions of tree space. Region 4 is a terrace with the highest likelihood. Region 2 is a subset of region 3 and both are terraces identified from two different partitioning schemes. Region 3 corresponds to a maximal partitioning scheme. }
\label{fig4}
\end{figure}

\subsection{3. Impact on confidence assessment: bootstrap proportions}

Terraces add ambiguity to phylogenetic tree inference directly by increasing the number of equally optimal solutions. In the next two sections we discuss how they also alter {\em estimates} of the quality of trees. The bootstrap  \citep{fel85} case is especially clear, as we demonstrate with detailed discussion of a simple example (Fig.~\ref{figBoot1}). In keeping with many results related to terraces, much of this can be predicted by reference to the stand, $\cS$, induced by the taxon coverage matrix and partition scheme for a given tree. Henceforth, we assume inference is with MP or ML-EUL and that results for stands carry over to terraces exactly. We simplify the problem further by letting sequence lengths within loci be so long that the only variation in bootstrap replicates has to be due to the existence of terraces . Specifically, let $T^*$ be a tree with four subtrees, X,Y,A,B, in a topology such that there is a bipartition given by XA$\mid$BY (Fig.~\ref{figBoot1}). Assume that X and Y are leaf taxa, but A and B are subtrees with $n_A$ and $n_B$ leaves respectively. Further assume that sites are all evolving on $T^*$ according to a model with ``reasonable'' parameter values such that, together with the assumption that the number of sites, $l$, is large enough, if complete taxon coverage were available, both ML-EUL or MP would infer the correct tree with $\sim100\%$ accuracy. Assume there is a partition, $\cP$, of data $D$ into two loci, 1 and 2.  The taxon coverage matrix, $C$ is such that taxa in two of the subtrees, X and Y, are sampled for both loci, but taxa in subtree A are only sampled for locus 1, and subtree B only for locus 2 (i.e., the same as in Fig.~\ref{fig2}, except no outgroup is present). 

\begin{figure}[!h]
\centering
\includegraphics[scale=0.66]{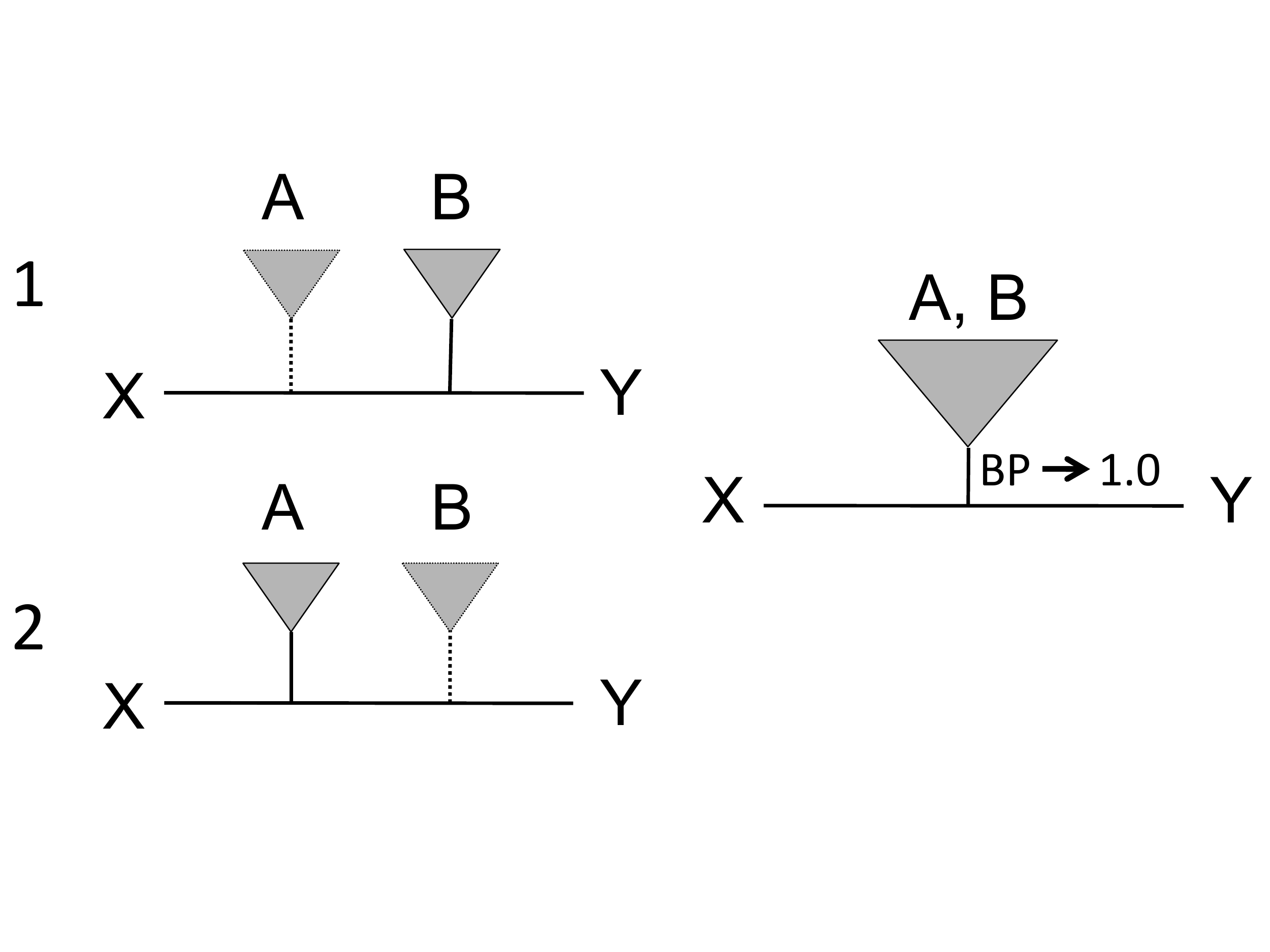}
\caption{ Impact of terraces on the bootstrap. True tree and taxon coverage as in Fig.~\ref{fig2} except outgroup removed and taxa $A$ and $B$  are now subtrees with $n_A$ taxa (present only for locus 1) and $n_B$ taxa (present for locus 2) respectively. As $n_A$ and $n_B$ get larger the bootstrap proportion for the incorrect bipartition at right goes to 100\%.}
\label{figBoot1}
\end{figure}

With complete sampling and ``ideal'' simulation conditions (``F84'' substitution model with equal base frequencies in Seq-Gen [\citealp{Ramb97}]; $l$ = 10,000 nt; all edge lengths = 0.1 substitutions/site; constant rates across sites), the bootstrap proportion (BP) for XA$\mid$BY are indeed 100\%, whether inferred with MP or ML-EUL . However, with the partial taxon coverage, $C$, outlined above, the true tree, $T^*$, occurs in a stand, $\cS$, with other trees. The number of trees in  $|\cS|$ grows with $n_A$ and $n_B$ (Fig.~\ref{figBoot2}), and many conflict with $T^*$ specifically by interleaving leaves from subtree A with leaves from subtree B, such that AB$\mid$XY is seen instead of XA$\mid$BY. In fact, perhaps surprisingly, these incorrect bipartitions rise in frequency in the set of bootstrap bipartitions with increasing $n_A$ and $n_B$ (Fig.~\ref{figBoot2}). For example, if $n_A = 3$ and $n_B = 3$, then the stand has 107 binary trees, but only 2 of these are consistent with the XA$\mid$BY bipartition in the true tree, $T^*$. Since topological error due to the sequence data has been factored out by long sequence length, bootstrapping will randomly sample the 107 trees on the terrace, and assuming one tree per replicate is kept (see below), the BP for the incorrect bipartition AB$\mid$XY will converge to 105/107 = 0.98. If $1-BP$ is regarded as a $P$-value for the null hypothesis of non-monophyly \citep{efr96,sus09}, this significance level is greatly inflated.

\newpage
			\begin{figure}[!h]
 \centering
 \includegraphics[scale=0.75]{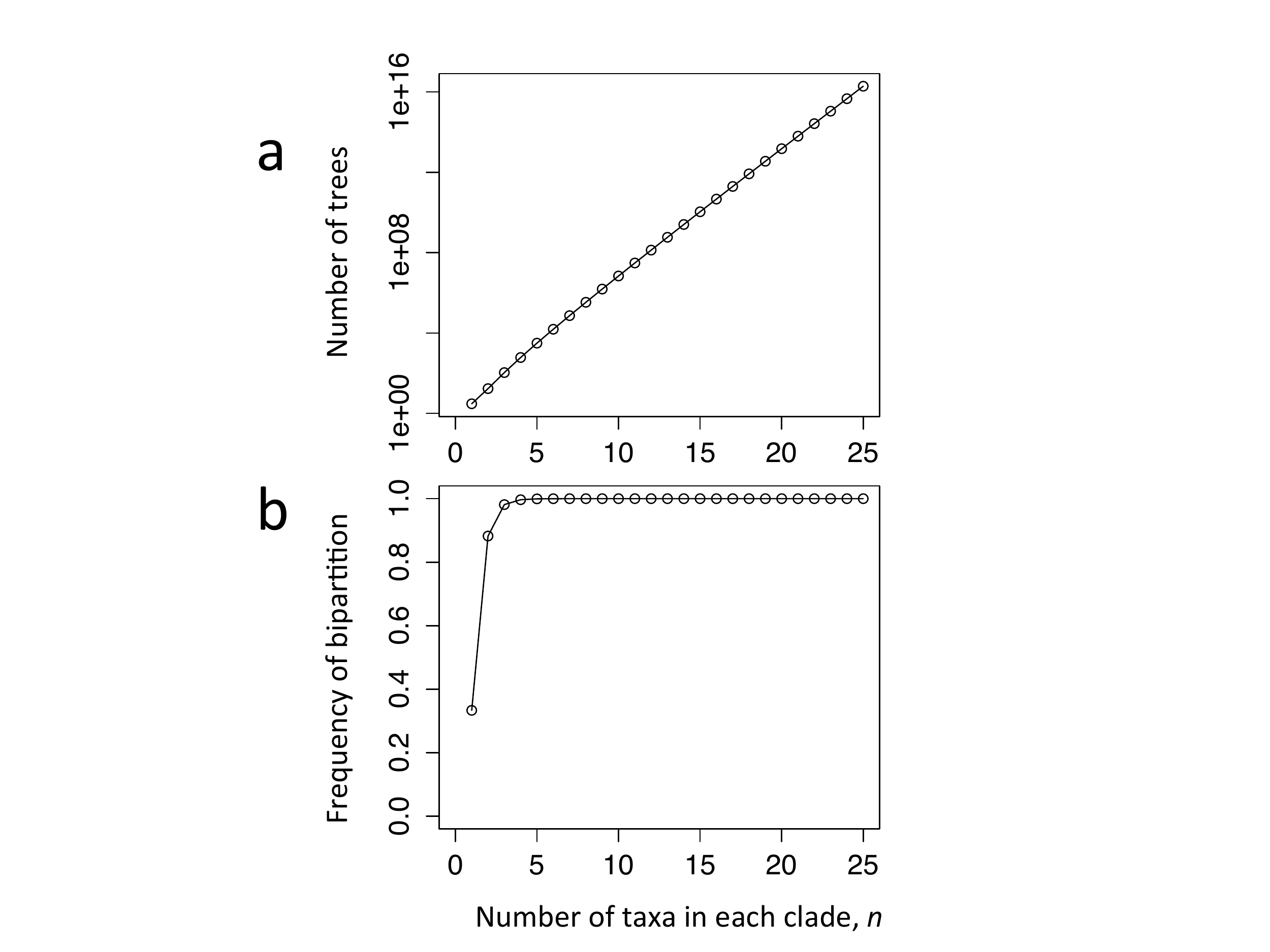}
 \caption{ Scaling of terrace size (A) and frequency of ``incorrect" AB$\mid$XY bipartition (B), in the example of Fig.~\ref{fig2}.  Clade size, $n$, is treated as the same in both A and B of Fig.~\ref{fig2}.}
 \label{figBoot2}
 \end{figure}

	To explore this small example more generally, we need to know how many rooted binary trees correctly display the subtrees $T^* | \cL_A$ and $T^* | \cL_B$, where we have changed notation a bit and $\cL_A$ refers to the subtree having just the leaf taxa in clade A, and likewise for B. Notice that the taxon sets $\cL_A$ and $\cL_B$ are disjoint. Let $n = n_A + n_B$, and $R(k) =1 \times 3 \times 5 \times \dots \times(2k-3)$ be the number of rooted binary trees for $k$ leaves.  If $n_A$ or $n_B \le 3$, then the number of trees on the terrace is  
	\begin{equation}
	\label{root_count}
	2+\frac{R(n)}{R(n_A)R(n_B)},
\end{equation}
a result that depends only on the two clade sizes. The ``2'' term corresponds to the two trees on the terrace with bipartitions of XA$\mid$BY and XB$\mid$AY.  The term on the right corresponds to the set of trees all of which have the bipartition AB$\mid$XY, but there are several ways to interleave the A and B clades and still display the original two subtrees. If both A and B are larger than three leaves, the expression depends on the topology of the subtrees (i.e., the entire tree), not just their sizes, a result due to \citet{con86}, so a more complex calculation for Equation~(\ref{root_count}) would be necessary. On {\em average}, however, across a uniform random sample of clade topologies, the mean number of trees converges to the expression in Equation~(\ref{root_count}). Thus, for certain patterns of partial taxon coverage, the bootstrap proportion for a given clade may tend to 0 or 1 depending not on the sequence data but on the number of taxa (and possibly tree topology) and taxon coverage pattern in specific clades in the tree.
	
These results implicitly assume that bootstrap protocols sample trees from a terrace uniformly at random. They may not. Heuristic search strategies may not be constructed to guarantee this. Moreover, the very notion of ``equality'' of the likelihood score in numerically intensive likelihood calculations is problematic, and maximum likelihood reconstruction in RAxML \citep{Stam14} and GARLI \citep{Zwickl06} operates on the implicit assumption that equally optimal trees either do not exist or cannot be unambiguously identified during tree search. For maximum parsimony inference, the situation is a little clearer. With integer value parsimony scores, maximum parsimony inference can certainly identify sets of trees, and there are then at least two choices for treating them. PAUP \citep{paup} uses a ``frequency within replicates''  \citep{dav04} approach in which the the frequency of a bipartition in the equally optimal trees found in a given replicate is summed across replicates. This procedure would effectively treat the trees on a terrace in the same way as our analytical results above. However, an alternative is a more conservative ``strict consensus'' approach in which a set of equally optimal trees must all have the bipartition for it to be regarded as present in that replicate \citep{dav04,SimFreu11}. Were we to adopt such a rule here, no bipartitions involving A, B, X, and Y would be supported, a very different result.

\subsection{4. Impact on confidence assessment: Bayesian posterior probabilities}
Given the issues raised by terraces for MP and ML-EUL, it is reasonable to expect some effects on calculation of Bayesian posterior probabilities, but in this case the impact is determined by an interaction between the taxon coverage pattern, tree topology, and priors on edge lengths. Let us simplify the bootstrap example to a tree with four leaves (i.e., $n_A = n_B =1$ in Figure~\ref{figBoot1}) but otherwise with the same taxon coverage matrix and partitioning scheme. Each element of the partition is 5000 nt, simulated in Seq-Gen \citep{Ramb97} under an F84 model with equal base frequencies, on the true tree AX$\mid$BY. We assume that likelihood calculations are carried out with an EUL model so stands are also terraces with respect to the likelihood score. There is then a stand of three binary trees, all with the same likelihood (in fact, these are all possible binary trees for four taxa). Perhaps surprisingly, however, their posterior probabilities may not be not equal. In particular, the posterior for the incorrect bipartition (AB$\mid$XY) increases as the lengths of the edges leading to A and B increase (calculated in MrBayes v. 3.1.2 [\citealp{ron12}]; 5 million generations; default burn in). For example, if the two long edges are four times the length of the other edges in the tree, the posterior probability is 81\% for the incorrect tree (Fig.~\ref{fig5}), whereas if the edge lengths are all the same, the correct tree has highest posterior probability.  

\begin{figure}[!h]
\centering
\includegraphics[scale=0.66]{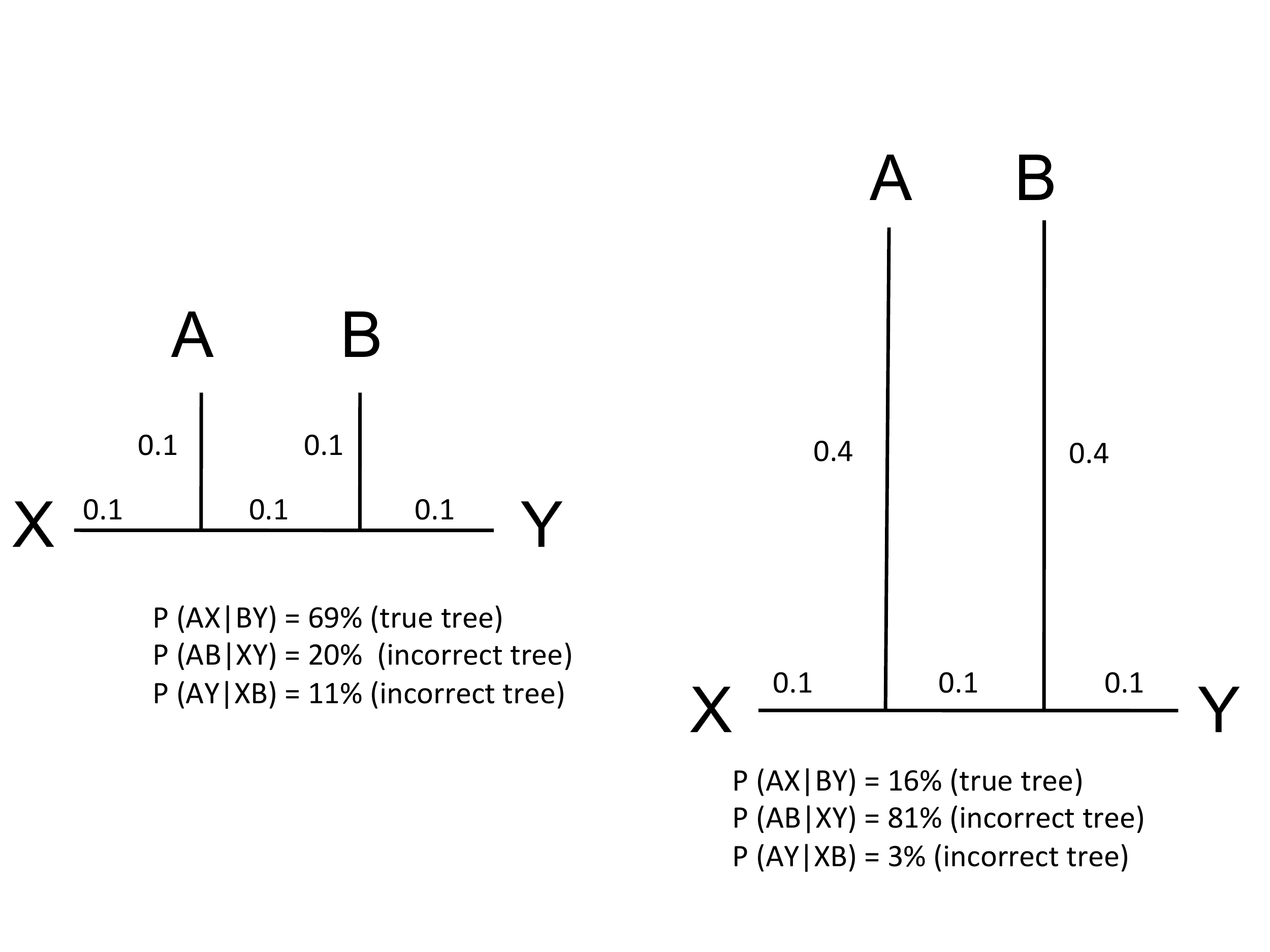}
\caption{ Impact of terraces on Bayesian posterior probabilities. True tree and taxon coverage as in Fig.~\ref{fig2} except outgroup removed. As terminal edge lengths for A and B get longer the posterior probability inferred by MrBayes (v. 3.1.2: 5 million generations; default burn in) for the incorrect bipartition, AB$\mid$XY increases. In this example each element of the partition is 5000 nt, simulated in Seq-Gen under a F84 model with equal base frequencies. Edge lengths are indicated in substitutions/site.}
\label{fig5}
\end{figure}  

This phenomenon can be traced ultimately to how Bayesian inference handles missing data. We can begin to understand this by considering the the simplest context of partial taxon coverage: with a trivial ``partition'' consisting of only a single block and having one leaf taxon, $x$, with missing data throughout that block. In other words, this is a data set in which one leaf, $x$, has no data at all. Consider a tree, $T$,  with branch lengths, containing taxon  
$x$ as a leaf, and the tree, $T_{-x}$ (with its inherited branch lengths) obtained by deleting leaf $x$ (and its incident branch) from $T$.    If the sequence of $x$ consists entirely of ``?''s, the parsimony and likelihood scores of any tree, $T'$, obtained by attaching $x$ to some edge of $T_{-x}$ is the same. In other words, a method such as maximum likelihood, which relies {\em only} on these scores is unable to decide the position of $x$. 

However, in Bayesian inference the position of $x$ may also be influenced by prior probabilities. In fact, we will now show that for the usual exponentially distributed prior on edge lengths used in phylogenetic inference, the posterior probability for different placements of $x$ in the tree is determined entirely by the length of the edge to which $x$ attaches; if these differ, then some placements of $x$ in this tree will have higher posterior probabilities. Ultimately this will affect posterior probabilities of trees on terraces in more complex partitioning schemes with partial taxon coverage.

Consider data generated by a reversible Markov process on a binary phylogenetic tree $T$ with branch lengths assigned, and analyzed under a Bayesian approach in which (i) all rooted binary phylogenetic trees have 
the same prior probability, and (ii)  an exponential prior applies independently across the branch lengths.
 Suppose we have a binary phylogenetic tree $T'$ that satisfies $T'_{-x} = T_{-x}$ (i.e. $T$ has the same or different placements of leaf $x$ in $T_{-x}$).  The following result describes the posterior probability of $T'$ (in the limit of long sequences) as the proportion of the total tree length of $T$ that comes from the edge of $T_{-x}$ to which $x$ would attach in order to produce $T'$.

\begin{theorem}
\label{thm2}
Let $T'$ be any binary phylogenetic $X$-tree  that agrees with $T$ up to the placement of taxon $x$, and consider a set of aligned sequences of length $k$ generated on $T$ with fixed branch lengths (under a standard reversible model) but with the sequence for taxon $x$ removed.
Then under conditions (i) and (ii) above, the Bayesian posterior probability of  $T'$ for this data converges in probability to $l(e_x)/\sum_{e }l(e)$ as the sequence length $k$ grows. Here, $l(e)$ is the length of edge $e$ in $T_{-x}$, the summation is over all edges of $T_{-x}$, and $e_x$ is the edge of $T_{-x}$ that $x$ must  attach to in order to produce the tree topology $T'$.

\end{theorem}

The proof of this result is presented in the Appendix. Intuitively, the theorem says that the relative chances of a leaf taxon attaching to two alternative edges of a tree are (for long sequences) given by their relative edge lengths: the longer edge will ``attract'' the taxon with missing data. This result may seem like something of a curiosity for a trivial partitioning scheme consisting of just a single locus, and a leaf taxon with no data, but it begins to explain why the posterior probability of trees on a terrace may be different even if their likelihoods are the same. In particular it pinpoints the important role that edge lengths can play in influencing subsets of the trees on a terrace to have higher posterior probabilities, as seen in the example in Figure ~\ref{fig5}. 

\subsection{5.  Extensions to linked models}

Our final result extends results on terraces to the setting in which there is some linkage among edge rate parameters across the loci in a data partition. We describe an example in which conditions are present such that terraces would occur for ML-EUL inference, but when an EL model is used instead, the accuracy of phylogenetic inference is negatively affected. Specifically, the relative likelihood scores can be re-ordered so that an incorrect tree is favored, much as in the Bayesian case described above.

 			\begin{figure}[!h]
\centering
\includegraphics[scale=0.66]{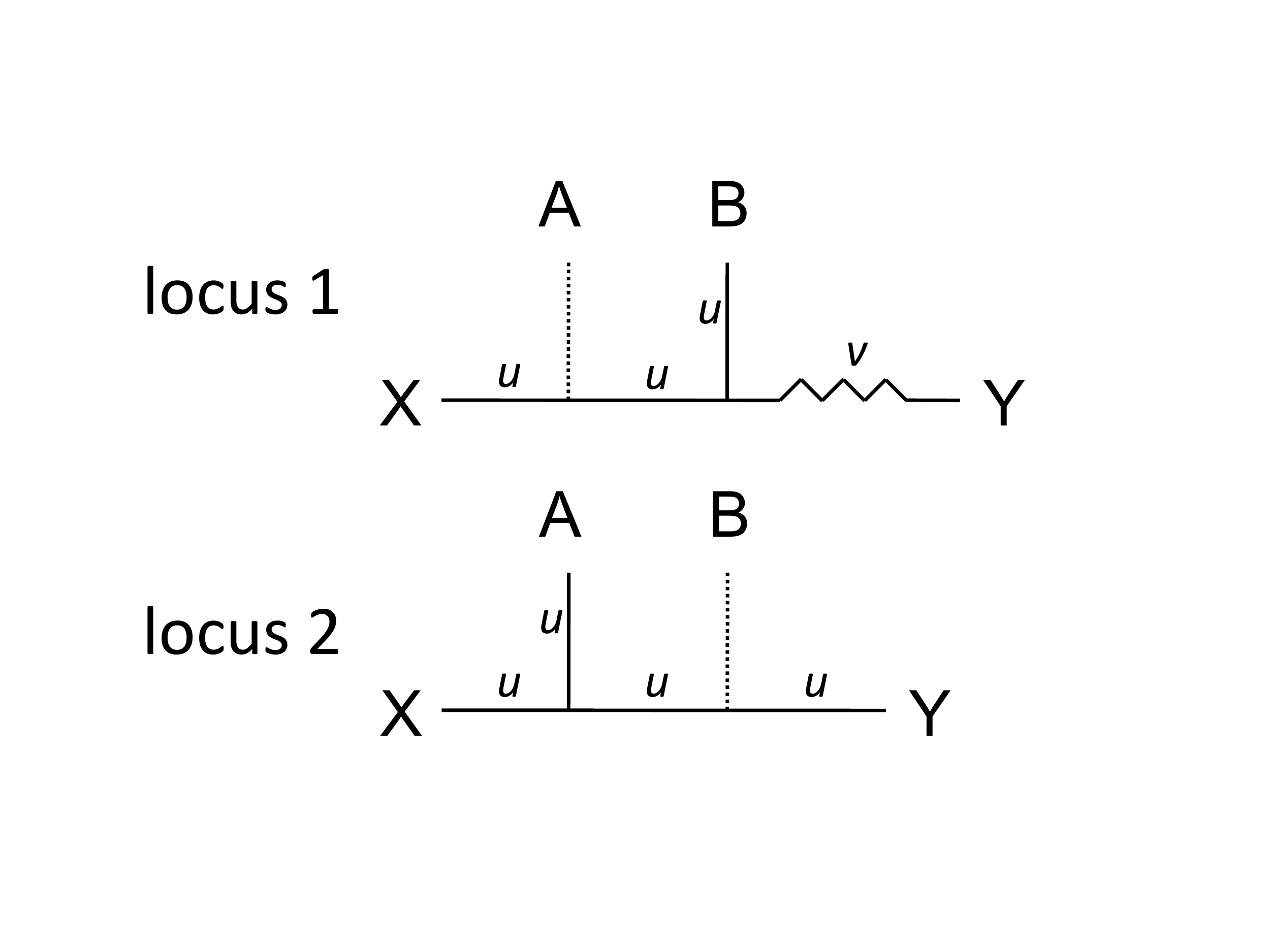}
\caption{ A simple model of heterotachy for two loci. Trees labeled 1 and 2 are annotated with edge length parameters for the two loci in the partition. Taxon coverage as in Fig.~\ref{fig2}. All edge lengths are $u$ for loci 1 and 2, except the outlier edge labeled $v$ for locus 1. }
\label{fig6}
\end{figure}

The example we describe involves sequences that exhibit heterotachy; that is, the pattern of edge length variation among edges is different between loci. We describe a simple case of this in a 4-taxon data set for two loci. Assume all edges for both loci have evolved with the same rate, $u=0.1$ substitutions/site, except for a single terminal edge for one of the loci, which has rate $v$, ranging from 0.1 to 1.0 substitution/site (Fig.~\ref{fig6}). Each locus in the partition has 5000 sites and sequences are evolved in Seq-Gen \citep{Ramb97} with the same rate matrix used in previous examples. The pattern of partial taxon coverage shown (Fig.~\ref{fig6}) induces a stand consisting of all three of the binary trees possible for four taxa. Thus, with MP or ML using EUL inference, the optimality scores of all three trees are the same. We examined two EL models. The first (``homotachy") model assumes the substitution rate matrix is the same for both loci and the edge rate parameters are the same for both loci; the second (``proportional") model assumes that the substitution rate matrix is the same for both loci but that edge rates for all edges at the second locus are strictly proportional to the corresponding edge rates for the first locus. This allows rate variation between loci, but is not ``heterotachy" per se. Likelihood calculations were carried out using PAUP* 4.0 (the ``proportional" model implemented with the ``siterates" command: \citealp{paup}). Under these conditions the likelihood scores for the three possible binary trees are different for both models (Fig.~\ref{fig7}). In particular, if $v$ is more than about $4u$, then an incorrect tree has the highest likelihood (Fig.~\ref{fig7}). In other words, use of an EL model for inference does not just avoid the ambiguity of the terrace phenomenon that would have been seen using EUL inference, it is positively misleading. Under these conditions, using an EUL model for inference and dealing with the resulting ambiguity might be a preferable (if more conservative) procedure.

			\begin{figure}[!h]
\centering
\includegraphics[scale=0.75]{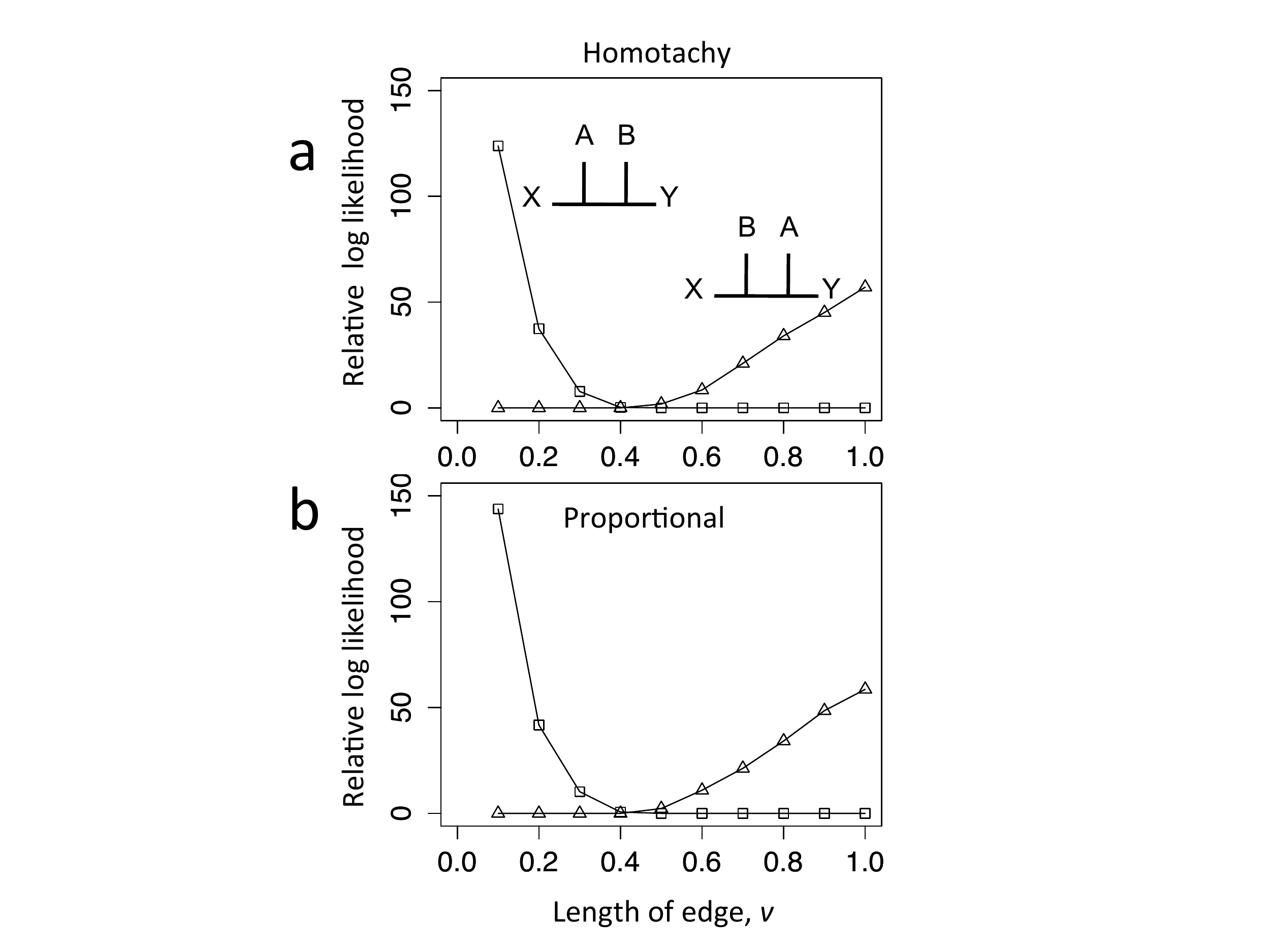}
\caption{ Likelihood scores of the best and second best trees relative to the worst tree as a function of edge length in the minimal model of heterotachy described in Fig.~\ref{fig6}. Worst tree is always the XY$\mid$AB tree. A. Inference model assumes the same edge parameters for both loci (homotachy). B. Inference model assumes edge parameters are 
proportional between loci (as implemented in the SiteRates option in PAUP* 4). Open squares refer to relative likelihood of AX$\mid$BY tree; open triangles refer to the (incorrect) BX$\mid$AY tree.}
\label{fig7}
\end{figure}  
 
This case clearly involves model misspecification, since the EL model described above is not the model actually generating the sequences. Interestingly, if we ``fix" the data set by filling in the missing entries with data generated under our heterotachous model, the correct ranking is restored under EL inference, even with this model misspecification. By the same token, eliminating heterotachy by setting $u=v$ in the generating model, but retaining missing data according to the specified partial coverage pattern, also restores the correct ranking of trees when using the EL model. Thus, the re-ranking of trees we observe in this example is not a product of either partial coverage alone or heterotachy alone. Both are required, and the correct ranking of the three trees' likelihood scores can be rescued by eliminating one or the other.

Although these results imply strongly that ML inference can be statistically inconsistent (i.e., converging to the wrong tree) under the combination of missing data and heterotachy, we have not been able to formally prove this in any sort of generality. However, some aspects of the converse are provable. For data evolving according to our heterotachy model, both parsimony and likelihood inference using a homotachy model are statistically consistent.  The proof that parsimony will be consistent for the heterotachy model with all taxa present, and the branch lengths indicated (for any $u, v$) follows from a straightforward application of Theorem 8.7.1 in \citet{sem03b}, which shows that parsimony is consistent for tree 1 and for tree 2 with those branch lengths (for any $u$ and also any $v$). Since parsimony is a linear
scoring scheme, if it is consistent on each block of a partition it is consistent on all the data. 

To prove the same for likelihood, consider first a single site -- if in any tree just one branch length is changed, say a pendant branch length from $u$ to $v$, and if $z$ is the state at this pendant leaf, $x$ the state at the other end (interior node) of this edge, and $Y$ is the collective set of states at the rest of the tree,
then by the Markov property, if proportion $\alpha$ of sites evolve under tree 1 (and $1-\alpha$ under tree 2) then for the heterotachy model we have

 \begin{equation}
\label{star1}
 Pr(Y,z|x) = \alpha Pr_u(Y,z|x) + (1-\alpha) Pr_v(Y,z|x).
\end{equation}
where $Pr_s(Y,z|x)$ denotes the conditional probability of leaf states $Y$ and $z$, given that interior node is $x$, and $s$ is the length of the pendant edge in question.

However, if we set the length of this pendant edge to $\theta$, where (for the 2-state symmetric model)

\begin{equation}
\label{star2}
\exp(-2\theta) = \alpha \exp(-2u) + (1-\alpha)\exp(-2v)
\end{equation}
then, from (\ref{star1})

$$Pr(Y,z|x)=Pr_\theta(Y,z|x)$$
for all $x,y, Z$, and so

$$Pr(Y,z)=Pr_\theta (Y,z)$$
and in particular, $Pr(Y,z)$ (=site pattern distribution under the heterotachy model) matches exactly the probability distribution of a homotachy model in which the length of the pendant edge is a particular value $\theta$ chosen  by (\ref{star2}) and lying somewhere between $u$ and $v$.  In other words, this heterotachy model produces site pattern probabilities the same as a homotachy model on the same topology and with an intermediate branch length.
These results are based on a single site, but it follows by Theorem 5.1 of \cite{cha96}, that if we now use ML (on sequences from the model) under the homotachy model this will be a statistically consistent estimator of tree topology.

\section{Discussion}

The new results on terraces presented here are relevant to an increasingly common setting for phylogenetic analysis in which large multi-locus data sets with significant numbers of missing sequences are combined with models that partition the data in various ways during likelihood or Bayesian inference. They are also relevant to analyses using maximum parsimony irrespective of partitioning schemes, which remains a computationally attractive method for analyses of large data sets (e.g., \citealp{Yoder13}). Previously we had shown that the terraces of equally optimal trees surrounding any specific tree could be quite large, although it was possible to exploit certain mathematical results on subtrees and supertrees to help characterize these sets of trees. In the present paper we found it useful to define ``stands" of trees and distinguish between stands and terraces. Stands are collections of trees that all display the subtrees associated with the separate blocks within a partitioned data set. Under some conditions the trees in a stand can all have the same optimality score, in which case they are called terraces, to reflect their constant ``elevation" along a vertical axis corresponding to that score.

If the score is the tree's parsimony score, a stand is always a terrace. The first new result we described specifies conditions under which this will also be true for the likelihood score. We showed that a sufficient condition was that the parameters describing edge lengths in different partition blocks must be independent (``edge unlinked"). Previously we had  also assumed parameters of the rate matrix to be unlinked \citep{san11}. This is sufficient but not necessary. Because practice in current phylogenetics of multi-locus super matrices ranges from using completely linked models to highly partitioned ones that unlink both rate matrices and edge parameters, impacts of these divergent strategies should be assessed. Our results indicate it is possible that edge unlinked models can  induce potentially large terraces around any given tree found in a reconstruction. This raises the deeper question of whether this ambiguity accurately reflects phylogenetic uncertainty or is best ameliorated by adjustments to the model (setting aside the possibility of simply filling in the missing data).

The second new result indicates that the stand associated with one partitioning scheme can actually be part of a larger stand under a different partitioning scheme. For parsimony, this leads to strong conclusions: it is possible to easily identify the extent of the largest terrace associated with any given partitioning scheme simply by constructing the so-called maximal partitioning scheme for the data set (a simple product of the distribution of missing data in the matrix), and then characterizing the terrace for that. This places a useful upper bound on the extent of the ambiguity associated with partial taxon coverage, and therefore should serve as the default terrace reported for any specific tree when parsimony is the optimality score.  

These results are fundamentally based on the properties of stands. However, in considering the optimality scores of the trees on a stand, it is worth remembering that even a ``maximal'' terrace may be imbedded in a yet larger collection of trees with the same optimality score owing to homoplasy, constant sites, or possibly other factors in the data not related to partial taxon coverage. The maximal number of trees on a terrace provides a lower bound on the actual number of equally parsimonious trees for $\cT$, and it can be calculated without reference to the sequence data proper, $D$. This may be relevant for tree search heuristics, which are clearly challenged by tree landscapes such as these \citep{gol}. For example, rather than doing tree rearrangements to find equally parsimonious trees, it should be possible to enumerate them directly. These can form ``seeds'' from which to continue searching using expensive tree score computations. This applies to likelihood inference as well, conditional on the partitioning scheme.

The remaining new results address various aspects of the impact of terraces on assessing tree accuracy. The third and fourth results refer to situations in which  bootstrap proportions and Bayesian posterior probabilities, respectively, are determined in large part not by the information con tent of the sequence data but by the pattern of partial taxon coverage. These happen for different reasons. In bootstrapping, the trees sampled from a terrace end up reflecting the frequencies of trees on that terrace, which in turn reflects the pattern of taxon coverage and tree shape. One consequence is that a clade can be highly supported in such an analysis because most trees on the terrace (all with equal likelihood scores) have that clade. It is possible to view this as either a ``feature" or a ``bug", but the bottom line is that one might not be expecting such sensitivity to arise from factors other than the sequence data. Simmons and colleagues \citep{SimFreu11,Sim12b,Sim14b} have noted a number of impacts of missing data generally on bootstrap estimates and have been critical of its application in this context, suggesting several strategies to check well supported clades for spurious support. We suspect some but not all of their observations in real data sets (e.g., \citealp{Sim14b}) are due to issues related to terraces, but other factors involving homoplasy and phylogenetic signal proper are no doubt also involved. 

Similarly, Bayesian inference can be heavily influenced by posterior probabilities not being the same across trees on a terrace (again, despite the likelihood being the same). In this case, the result is largely explained by the influence of priors. With standard exponentially distributed prior probabilities on edge lengths, the probability that a single taxon with all its data missing will attach to a particular edge of the tree is determined by that edge length in relation to the length of the entire tree. Absent any other information, Bayesian inference will place a taxon on the longest edge of the tree (though not necessarily with high probability). Again, this may be viewed as desirable or not, but it has the downstream consequence that trees in a stand can have very different posterior probabilities, which are determined by a fairly opaque convolution of the priors, edge lengths, partial coverage patterns, and combinatorics of terraces and tree shapes.

A reasonable rejoinder to some of these concerns is that one should simply avoid edge-unlinked partitioning schemes in likelihood or Bayesian inference (and maximum parsimony) and enforce some more homogeneous model during statistical inference. Of course, this kind of underparameterization may cause its own problems, but we pursued briefly whether it at least can ameliorate some of the issues we have described with terraces. The results presented in the last section suggest reasons to proceed carefully. When sequence data are generated on a tree with even a small amount of heterotachy between loci, an edge-linked model is used for inference, and there is only partial taxon coverage, then the likelihood score for the correct tree can actually be worse than for an incorrect tree. Simulations with long sequence lengths imply this is an instance of statistical inconsistency. As stated this might seem like simply another case in which model misspecification causes problems, but in this instance, it is clearly a negative interaction between model misspecification and missing data, because the problem can be avoided by fixing either aspect. 

Much further work is needed to narrow down the precise effects of terraces on confidence estimation and accuracy of inference. In the meantime, however, there are options for remediation. Setting aside the strategy of acquiring the missing sequences to eliminate partial taxon coverage entirely, which may be expensive or impossible depending on availability of DNA samples, there are computationally promising approaches. Previously we posed the ``maximum defining label set" (MDLS) problem \citep{san11}, in which we seek the smallest number of leaf taxa to delete from the coverage matrix such that the stand for a given tree is reduced to a single tree. This has an exact and efficient solution for two loci, and experiments with data sets indicate there are interesting instances in which elimination of relatively few taxa can solve the problem, even while leaving a significant amount of missing data. However, there is no known exact solution to the case of three or more loci. The good news is that simple heuristics in the case of three or more loci can find solutions that eliminate terraces \citep{san11}; they  just may not be optimal (it may have been possible to do the same thing and keep more taxa in the matrix).   

To highlight the impacts clearly, most of our results were in the context of sequence data sets so large that the only error was due to partial taxon coverage. In real data sets there is also error from the finite sample taken from the substitution process. This translates into a broadening of the bootstrap or posterior distribution of trees. In addition, there may be distinct terraces associated with each sample tree taken from these distributions; and there may be a distinct MDLS solution for deleting some set of taxa for each of these trees. How do we integrate across this information to make headway in reducing the overall impact of terraces?  A simple but conservative fix might be to replace any sampled tree in a bootstrap replicate or an MCMC run with the strict consensus of the maximal stand in which that tree is imbedded. Then any clade on that tree is present in all trees on the terrace. Although this would tend to reduce the false positive clades uncovered in an analysis, it might lack sensitivity. Clades present at an overwhelmingly high frequency on the terrace but not quite 100\% will be missed by a strict consensus.

Another approach would be to rely on other sorts of summaries about trees than consensus. For example, a terrace could be characterized by the average dissimilarity among its trees, based on a measure of distance between trees, such as the Robinson-Foulds (RF: \citealp{RF81}) distance. Then, perhaps a more synthetic assessment of a confidence set of trees sampled from bootstrap replicates or posterior distributions could be to explore the ``distances" between terraces, using a measure of distance between sets, such as the Hausdorff distance \citep{Yu14}. The Hausdorff distance is small when each tree on one terrace is close (measured here by RF distance) to some tree on the other terrace. This might be true even if the average RF distance between trees {\em within} a terrace is much larger. Whether this approach ultimately proves promising or not, some means to characterize more fully the relationships between entire sets of trees seems to be a necessity when terraces are commonplace.

\section{Acknowledgments} 
This work was supported by the US National Science Foundation (DEB-1353815 to MJS, DJZ, AS, MMM, and MS) and the New Zealand Marsden Fund (to MS). We thank J. Charboneau and M. Simmons for discussion.

\bibliographystyle{sysbio}
\bibliography{TerracesMs}

\newpage

\section{Appendix: Mathematical proof of Proposition~\ref{thm1} and Theorem~\ref{thm2}}

\subsection{Proof of Proposition~\ref{thm1}}

Consider an analysis where for each locus $i$ we are free to select $E_i$-parameters, but the $M_i$ parameters are constrained to be identical (i.e. $M_i = M$ for all $i$).
Let $\varphi(T)$ denote the log-likelihood of tree $T$ (having leaf set $\cL$) for the data $D$.  
Then assuming the loci evolve independently (conditional on the parameter choices) we have:
$$\varphi(T) = \sup_{(M, (E_i))} \sum_{i=1}^k \log \PP(D_i| T, M, (E_i)),$$
where  $k$ is the number of loci, `sup' refers to supremum (i.e. maximum if it is attained, else its limiting value) as we search over $M$ and the $E_i$ parameter spaces,  and where 
$(E_i)$ is short for $(E_1, \ldots, E_k)$. 
Now, $$\PP(D_i| T, M, (E_i)) = \PP(D_i|(T|\cL_i), M, E_i),$$
(notice that $T$ and $(E_i)$ on the left has been replaced by $(T|\cL_i)$ and $E_i$ on the right).
Combining these last two equations gives:
$$\varphi(T) = \sup_{(M, (E_i))} \sum_{i=1}^k \log \PP(D_i|(T|\cL_i), M, E_i),$$
and so 
\begin{equation}
\label{var}
\varphi(T) = \sup_{M} \sum_{i=1}^k \sup_{E_i} \log \PP(D_i|(T|\cL_i), M, E_i).
\end{equation}
Now, suppose that $T$ is a phylogenetic tree on the entire  leaf set $\cL$ and that $T$
maximises $\varphi(*)$. Let $T'$ be any other phylogenetic tree on leaf set $\cL$ for which $T|\cL_i = T'|\cL_i$ (i.e. $T'$ lies in the same terrace at $T$).
Then from Eqn.~(\ref{var}) we have:
$$\varphi(T) =\sup_{M} \sum_{i=1}^k \sup_{E_i} \log \PP(D_i|(T|\cL_i), M, E_i) =$$
$$ \sup_{M} \sum_{i=1}^k \sup_{E_i} \log \PP(D_i|(T'|\cL_i), M, E_i) = \varphi(T'),$$
so $T'$ is an ML tree also.
In other words, all trees on the same terrace as $T$ are ML trees. This completes the proof.
\hfill$\Box$

\subsection{Proof of Theorem~\ref{thm2}}

We first begin by defining more formally some of the notions mentioned earlier.

\begin{itemize}
\item
Given any binary phylogenetic $X$ tree $T'$, and any  taxon $x$ from $X$, consider the tree $T'_{-x}$ that is obtained from $T'$ by deleting leaf taxon $x$ and its incident edge $e(x)$.  Note that each edge of $T'_{-x}$ corresponds to an edge of $T'$, except for the edge $e_{-x}$ of $T'_{-x}$ which corresponds to the two edges $(e_1, e_2)$ of $T'$, that are incident with $e(x)$ in $T'$.  In this way, if $T'$ comes equipped with an branch length assignment $l$ (so $l(e)$ is the length of edge $e$), then the induced branch length $l_{-x}$ function for $T'_{-x}$ is thus given by:
$$l_{-x}(e) = 
\begin{cases}
l(e), & \mbox{ if } e\neq e_{-x}; \\
l(e_1)+l'(e_2), & \mbox{ if } e=e_{-x}.
\end{cases}
$$

\item For any data set $D$ that consists of a sequence of $k$ aligned site patterns on $X$, and any taxon $x \in X$, let $D_{-x}$ denote sequence of $k$ aligned site patterns on $X-\{x\}$ obtained by deleting the sequence for $x$.

\item Now suppose that the sequence sites in $D$ have evolved i.i.d. on some fixed binary phylogenetic $X-$tree $T$ with branch length assignment $\lambda$,
under a reversible Markovian process.    Thus the sites in $D_{-x}$ evolve i.i.d. on $T_{-x}$ with branch length assignment $\lambda_{-x}$.

\end{itemize}

We wish to apply a Bayesian approach to compare different placements of the taxon $x$ into $T_{-x}$ given the censored data $D_{-x}$.  We assume a prior probability distribution on the set of binary phylogenetic $X$-trees with branch lengths for which:
\begin{itemize}
\item[(i)] each binary tree has the same probability (i.e. the `PDA distribution');
\item[(ii)] edge lengths are independent exponential random variables.
\end{itemize}
Without loss of generality (by rescaling) we may assume that the mean of the exponential distribution in (ii) is 1.

Suppose we have two binary phylogenetic $X$-trees $T'$ and $T''$ that satisfy $T'_{-x} = T''_{-x} = T_{-x}$ (i.e. two different placements of leaf $x$ in $T_{-x}$).  One (or neither) of these trees
might be $T$.  We are interested in the ratio of posterior probabilities $\frac{\PP(T'|D_{-x})}{\PP(T''|D_{-x})}$. The following result states that for long sequences this ratio converges towards the ratio of the lengths of the two edges of $T_{-x}$ to which the missing taxon ($x$) is attached. To establish Theorem 4, we prove the following result. 

\begin{theorem}
\label{thm2a}
For data generated by a reversible Markov process on a phylogenetic $X$-tree $T$ with branch length assignment  $\lambda$, consider, for any $x \in X$, any two phylogenetic $X$-trees $T'$ and $T''$ obtained by attaching $x$ to edges of $T_{-x}$ of length $l'$ and $l''$ respectively. Then the ratio $\frac{\PP(T'|D_{-x})}{\PP(T''|D_{-x})}$ converges in probability to $l'/l''$,
as  the sequence length $k$ grows. 
\end{theorem}
Notice that Theorem~\ref{thm2a} implies Theorem~\ref{thm2} since the statistical consistency of Bayesian phylogenetics under identifiable models implies that  any tree which is different from $T_{-x}$ when $x$ is deleted has a posterior probability that converges to zero as the sequences length $k$ grows (thus $\PP(T''|D_{-x})$ sums to 1 as we sum over all binary trees $T''$ that agree with $T$ up to the placement of $x$).
Thus, the the remainder of our argument is tailored towards proving Theorem~\ref{thm2a} under the same conditions stated for Theorem~\ref{thm2} (in particular, conditions (i) and (ii) in the preamble to that theorem). 

From Bayes' identity we have:
\begin{equation}
\label{eq1}
\PP(T'|D_{-x}) = \frac{\PP(D_{-x}|T')\PP(T')}{\PP(D_{-x})}.
\end{equation}

Now, $\PP(T)=\PP(T')$ (by assumption (i)), and so, from Eqn. (\ref{eq1}) and the analogous identity for $\PP(T|D_{-x})$:
\begin{equation}
\label{eq2}
\frac{\PP(T'|D_{-x})}{\PP(T|D_{-x})} = \frac{\PP(D_{-x}|T')}{\PP(D_{-x}|T)}.
\end{equation}

Moreover,  
\begin{equation}
\label{eq3}
\PP(D_{-x}|T) = \int_\Gamma \PP(D_{-x}|T_{-x}, l) f(l) dl
 \end{equation}
 where $l$ is the set of branch length assignment on $T_{-x}$, and where $f(l)$ refers to the density of the branch lengths on $T_{-x}$ that is induced by independent exponential prior branch lengths on $T$
 ($\Gamma$ is the set of possible branch lengths of $T_{-x}$). 
 
 Similarly,
 \begin{equation}
\label{eq4}
\PP(D_{-x}|T') = \int_\Gamma \PP(D_{-x}|T'_{-x}, l) f'(l) dl
 \end{equation}
  where $l$ is the branch length assignment on $T'_{-x} (=T_{-x})$, and where $f'(l)$ refers to the density of the branch lengths on $T_{-x}$ induced by the independent exponential priors on $T'$.

 Let $\hat{s}$ denote the empirical frequency distribution of site patterns on $X-x$, and let $p(l)$ denote the 
vector of site pattern probabilities generated by $T_{-x}$ with branch lengths $l$.
We have the identity:
 \begin{equation}
\label{eq5}
\PP(D_{-x}|T_{-x}, l)/\prod_i \hat{s}_i^{\hat{s_i}k} = \exp(-k d_{KL}(\hat{s} ||p(l)),
\end{equation}
where $i$ ranges over all cite patterns, and where $d_{KL}(P||Q)=\sum_i P_i\log(P_i/Q_i)$ refers to Kullback-Leibler separation of probability distributions $P$ and $Q$.
Similarly, 
 \begin{equation}
\label{eq5a}
\PP(D_{-x}|T'_{-x}, l)/\prod_i \hat{s}_i^{\hat{s_i}k} = \exp(-k d_{KL}(\hat{s} ||p( l)).
\end{equation}

Combining Eqns. (\ref{eq2}), (\ref{eq3}), (\ref{eq4}), (\ref{eq5}) and (\ref{eq5a}) we obtain:

 \begin{equation}
\label{eq6}
\frac{\PP(T'|D_{-x})}{\PP(T|D_{-x})} = \frac{\int_\Gamma \exp(-k d_{KL}(\hat{s} ||p(l)) f'(l)dl}{\int_\Gamma \exp(-k d_{KL}(\hat{s} ||p(l)) f(l)dl}.
\end{equation}

Now,  let $B_k$ denote the subspace of the branch length space $\Gamma$ of $T_{-x}$ that is within ($l_\infty$) distance   $k^{-1/4}$ of $\lambda_{-x}$. 
Then  for $g=f$ or $g=f'$ we have the following convergence in probability as $k$ grows:
\begin{equation}
\label{limit4}
R_k:=\frac{\int_{\Gamma-B_k}\exp(-kd(\hat{s}||p(l))g(l)dl}{\int_{B_k}\exp(-kd(\hat{s}||p(l)) g(l)dl} \xrightarrow{p} 0.
\end{equation}
The proof of this last equation is given in a separate subsection below.  We apply it as follows. Notice that for  $g=f$ or $g=f'$, we have:
$$\int_{\Gamma}\exp(-kd(\hat{s}||p(l))g(l)dl = \int_{B_k}\exp(-kd(\hat{s}||p(l))g(l)dl + \int_{\Gamma-B_k}\exp(-kd(\hat{s}||p(l))g(l)dl,$$
and so
\begin{equation}
\label{lim1}
\int_{\Gamma}\exp(-kd(\hat{s}||p(l))g(l)dl = (1+ R_k)\int_{B_k}\exp(-kd(\hat{s}||p(l))g(l)dl .
\end{equation}
Moreover, since $g$ is continuous,  and the nested sequence of sets $B_k$ convergences on the vector $\lambda_{x}$ as $k \rightarrow \infty$ we have:
\begin{equation}
\label{lim2}
\frac{\int_{B_k}\exp(-kd(\hat{s}||p(l))g(l)dl }{\int_{B_k}\exp(-kd(\hat{s}||p(l))dl}  \xrightarrow{p}  g(\lambda_{-x}),
\end{equation}
as $k$ grows.
Thus, combining Eqns. (\ref{eq6}), (\ref{lim1}) and (\ref{lim2}) we obtain:
\begin{equation}
\label{eq6a}
\frac{\PP(T'|D_{-x})}{\PP(T|D_{-x})}    \xrightarrow{p} \frac{f'(\lambda_{-x})}{f(\lambda_{-x})}.
\end{equation}

Now, by assumption (ii),  the branch lengths in $T_{-x}$ are independent exponentials (of mean 1)  for all edges other than $e_{-x}$, and for this edge the branch length is the sum of two independent exponential(s) of mean 1, which has a gamma distribution  with density $f(t) =t\exp(-t)$. Thus 
 \begin{equation}
\label{eq7}
f(l) = \prod_{e'\neq e_{-x}} \exp(-l(e))\cdot [ l(e_{-x})\exp(-l(e_{-x}))] = l(e_{-x}) \prod_{e} \exp(-l(e)),
\end{equation}
where the last product term is over all edges of $T_{-x}$.
Similarly, if $e'_{-x}$ is the edge of $T_{-x} (=T'_{-x})$ that $x$ is attached to in $T'$ then 
 \begin{equation}
\label{eq8}
f'(l) = l(e'_{-x}) \prod_{e} \exp(-l(e)).
\end{equation}
From Eqns. (\ref{eq7}) and (\ref{eq8}) we have:
$$\frac{f'(\lambda_{-x})}{f(\lambda_{-x})} = \frac{\lambda_{-x}(e'_{-x})}{\lambda_{-x}(e_{-x})},$$
which, from (\ref{eq6a}), implies that for $T'$ and $T''$ with $T'_{-x} = T''_{-x} = T_{-x}$, we have:
$$\frac{\PP(T'|D_{-x})}{\PP(T''|D_{-x})}    \xrightarrow{p} \frac{l'}{l''},$$ where 
$l'=\lambda_{-x}(e'_{-x})$ and $l''=\lambda_{-x}(e''_{-x})$, and where  $e'_{-x}$ and $e''_{-x}$ are the corresponding edges of $T_{-x}$ that $x$ attaches to in $T'$ and $T''$ respectively. 
This completes the proof of Theorem~\ref{thm2a} and thereby Theorem~\ref{thm2}, modulo the remaining step of establishing Eqn.~(\ref{limit4}) which we attend to below. 
\hfill$\Box$

\subsection{Proof of Eqn.~(\ref{limit4})}

A classic result (e.g. Wilk's theorem) ensures that the following convergence in distribution holds: 
$$
2k d(\hat{s}||p(\lambda_{-x})) \xrightarrow{D} \chi^2_{N-1}
$$
where $\chi^2_{N-1}$ is a chi-square distribution with $N-1$ degrees of freedom (here $N$ is the number of possible site patterns).
By the continuous mapping theorem, it now follows that:
\begin{equation}
\label{limit}
\exp(-k d(\hat{s}||p(\lambda_{-x}))) \xrightarrow{D} W
\end{equation}
where $W= \exp(-\chi^2_{N-1})$ is a continuous and non-negative random variable. 

Moreover, if a sequences of branch length vectors $l_k$ lies within ($l_\infty$) distance $\frac{1}{k}$ of $\lambda_{-x}$ then
we also have:
\begin{equation}
\label{limit2}
\exp(-k d(\hat{s}||p(l_k))) \xrightarrow{D} W.
\end{equation}
(For further details see \citet{ser80}, esp. Section 3.5).

Next, Pinsker's inequality (see \citet{cov06}) gives for any $l \in \Gamma$:
$$d(\hat{s}||p(l)) \geq \frac{1}{2}||\hat{s} - p(l)||_1^2,$$ where $\|\cdot||_1$ refers to the $l_1$ metric.
The triangle inequality for this metric then gives:
\begin{equation}
\label{h1}
d(\hat{s}||p(l))\geq \frac{1}{2}(\|p(\lambda_{-x}) -p(l)\|_1 - \|\hat{s} - p(\lambda_{-x})\|_1)^2.
\end{equation}
Now, for any $l \in \Gamma-B_k$,  we have $\|l - \lambda_{-x}\|_\infty \geq k^{-1/4}$, and so, by Theorem 2.1(2) of 
\cite{mou99} there exists a pair of leaves $i,j$ so that the difference in path length between these leaves under branch lengths $l$ and $\lambda_{-x}$ is at least
$\frac{1}{2} k^{-1/4}$.  Since the site substitution model is reversible, the probability two leaves are in the same state is a monotone decreasing function of the path length between
them (a positive mixture of exponential functions).  This in turn implies that the event that leaves $i$ and $j$ are in the same state 
differs in probability under the branch lengths $l$ and $\lambda_{-x}$ by an amount that is at least $k^{-1/4}$ times some constant (dependent on the model, and $\lambda_{-x}$). 
In particular,
\begin{equation}
\label{h2}
\|p(\lambda_{-x}) -p(l)\|_1 \geq c k^{-1/4}, \mbox{  for some constant } c>0.
\end{equation}
Also, 
\begin{equation}
\label{h3}
\|\hat{s} - p(\lambda_{-x})\| \leq k^{-1/3}, 
\end{equation} with probability converging to $1$ as $k$ grows. Consequently, by combining Eqns. (\ref{h1}), (\ref{h2}) and (\ref{h3}), 
the following inequality
holds for all $l \in \Gamma-B_k$ with probability converging to $1$ as $k$ grows:
$$\frac{\exp(-kd(\hat{s}||p(l)))}{\exp(-dk^{1/2})} \leq 1,$$
for some constant $d>0$. 
Thus, with probability converging to 1 as $k$ grows: 
\begin{equation}
\label{fun}
\frac{\int_{\Gamma-B_k}\exp(-kd(\hat{s}||p(l))g(l)dl}{\exp(-dk^{1/2})} \leq  \lim_{k\rightarrow \infty} \int_{\Gamma-B_k} 1 \cdot g(l)dl =1.
\end{equation}
On the other hand, if we let $B^*_k \subset B_k$ be the set of branch length vectors that lie within ($l_\infty$) distance at most $1/k$ from $\lambda_{-x}$ then 
\begin{equation}
\label{fun2}
\frac{\int_{B^*_k}\exp(-kd(\hat{s}||p(l))g(l)dl}{\int_{B^*_k} g(l)dl} \geq \sup_{l \in B^*_k} \{ \exp(-kd(\hat{s}||p(l))\}.
\end{equation}

Now, for a value $C>0$ that is independent to $k$ (but dependent on $\lambda_{-x}$) we have $\int_{B^*_k} g(l)dl \geq C k^{-E}$, where $E$ is the number of edges
of $T$.
Also,  from Eqn.~(\ref{limit2}) 
$ \sup_{l \in B^*_k} \{ \exp(-kd(\hat{s}||p(l))\}$ converges in distribution to the random variable $W= \exp(-\chi^2_{N-1})$  as $k$ grows. 
 Thus, from (\ref{fun2}),   the following inequality holds with probability converging to 1 as $k$ grows:
 $ \int_{B^*_k}\exp(-kd(\hat{s}||p(l))g(l)dl/ C k^{-E} \geq W,$
 and thus, 
\begin{equation}
\label{fun3}
\frac{ \int_{B_k}\exp(-kd(\hat{s}||p(l))g(l)dl}{ k^{-E}} \geq W,
\end{equation}
since $B^*_k \subset B_k$ and the integrand is non-negative.
Combining Eqns.~(\ref{fun})  and (\ref{fun3}) the following inequality holds with probability converging to 1 with increasing $k$:
$$
R_k \leq \frac{\exp(-dk^{1/2})}{C k^{ -E}}\cdot \frac{1}{W}, 
$$
Notice that the second term on the right ($\frac{1}{W}$) is a continuous random variable, but since
$\PP(W>0)=1$ and since the first term on the right converges to zero (absolutely) as $k$ tends to infinity, this suffices to establish Eqn.~(\ref{limit4}).

\end{document}